\documentclass[sigconf]{acmart}
\AtBeginDocument{%
  }

\usepackage{bm}
\usepackage{algorithm, pseudocode}

\def\G {\ensuremath{\mathbf{G}}}

\def\K {\ensuremath{\mathbb{K}}}

\def\kbar {\ensuremath{{\overline{k}}}}
\def\qbar {\ensuremath{{\overline{q}}}}
\def\wbar {\ensuremath{{\overline{w}}}}

\def\calA {\ensuremath{\mathcal{A}}}
\def\calB {\ensuremath{\mathcal{B}}}
\def\calC {\ensuremath{\mathcal{C}}}

\def\scrR{\ensuremath{\mathscr{R}}}
\def\scrS{\ensuremath{\mathscr{S}}}

\DeclareBoldMathCommand{\A}{A}
\DeclareBoldMathCommand{\B}{B}
\DeclareBoldMathCommand{\C}{C}
\DeclareBoldMathCommand{\F}{F}
\DeclareBoldMathCommand{\G}{G}
\DeclareBoldMathCommand{\H}{H}
\DeclareBoldMathCommand{\J}{J}
\DeclareBoldMathCommand{\M}{M}
\DeclareBoldMathCommand{\L}{L}
\DeclareBoldMathCommand{\P}{P}
\DeclareBoldMathCommand{\Q}{Q}
\DeclareBoldMathCommand{\R}{R}
\DeclareBoldMathCommand{\V}{V}
\DeclareBoldMathCommand{\W}{W}
\DeclareBoldMathCommand{\U}{U}
\DeclareBoldMathCommand{\T}{T}
\DeclareBoldMathCommand{\X}{X}
\DeclareBoldMathCommand{\Z}{Z}
\DeclareBoldMathCommand{\Y}{Y}

\DeclareBoldMathCommand{\a}{a}
\DeclareBoldMathCommand{\e}{e}
\DeclareBoldMathCommand{\d}{d}
\DeclareBoldMathCommand{\f}{f}
\DeclareBoldMathCommand{\g}{g}
\DeclareBoldMathCommand{\h}{h}
\DeclareBoldMathCommand{\i}{i}
\DeclareBoldMathCommand{\j}{j}
\DeclareBoldMathCommand{\k}{k}
\DeclareBoldMathCommand{\p}{p}
\DeclareBoldMathCommand{\q}{q}
\DeclareBoldMathCommand{\t}{t}
\DeclareBoldMathCommand{\u}{u}
\DeclareBoldMathCommand{\w}{w}
\DeclareBoldMathCommand{\x}{x}
\DeclareBoldMathCommand{\y}{y}
\DeclareBoldMathCommand{\m}{m}
\DeclareBoldMathCommand{\r}{r}
\DeclareBoldMathCommand{\z}{z}
\DeclareBoldMathCommand{\v}{v}

\DeclareBoldMathCommand{\bell}{\ell}
\DeclareBoldMathCommand{\bcalC}{\calC}
\DeclareBoldMathCommand{\bcalA}{\calA}
\DeclareBoldMathCommand{\bcalB}{\calB}

\def\rho {\ensuremath{\varepsilon}}

\def\softO{\ensuremath{{O}{\,\tilde{ }\,}}}

\DeclareBoldMathCommand{\bcardA}{\mathcal{A}}
\DeclareBoldMathCommand{\bcardB}{\mathcal{B}}

\setcopyright{acmlicensed}
\copyrightyear{2026}
\acmYear{2026}
\setcopyright{cc}
\setcctype{by}
\acmConference[ISSAC '26]{51st International Symposium on Symbolic and Algebraic Computation}{July 13--17, 2026}{Oldenburg, Germany}
\acmBooktitle{51st International Symposium on Symbolic and Algebraic Computation (ISSAC '26), July 13--17, 2026, Oldenburg, Germany}
\acmDOI{10.1145/3815436.3815453}
\acmISBN{979-8-4007-2595-1/2026/07}

\begin{document}


\title{A Symbolic Homotopy Algorithm for Solving Composable Polynomial Systems}


\author{Thi Xuan Vu}
\email{thi-xuan.vu@univ-lille.fr}
\orcid{0000-0002-2285-7801}
\affiliation{%
  \institution{Univ. Lille, CNRS, Centrale Lille, UMR 9189 CRIStAL}
  \postcode{F-59000}
  \city{Lille}
  \country{France}
}

\renewcommand{\shortauthors}{Thi Xuan Vu}

\begin{abstract}
  A clear and well-documented \LaTeX\ document is presented as an
  article formatted for publication by ACM in a conference proceedings
  or journal publication. Based on the ``acmart'' document class, this
  article presents and explains many of the common variations, as well
  as many of the formatting elements an author may use in the
  preparation of the documentation of their work.
\end{abstract}

\begin{abstract}
  We study the problem of computing the isolated regular solutions of a system
\((f_1,\ldots,f_n)\) of \(n\) polynomial equations in \(n\) variables \((X_1, \dots, X_n)\) over a field of characteristic zero \(k\). 
We focus on systems with a \emph{composable structure}, where each polynomial \(f_i\) can be expressed as a composition
\(
f_i = h_i(g_1,\dots,g_n).
\) 
Exploiting this structure allows us to reduce the original system to one in the \(g_j\) variables, 
thereby significantly improving the efficiency of symbolic solution algorithms. 
We present a probabilistic algorithm that computes all isolated regular solutions, with arithmetic complexity being polynomial in the input size and in the number of solutions. 

A first important application is when \(f_1, \dots, f_n\) belong to the subring \(k[g_1, \dots, g_n]\), where \(g_1, \dots, g_n\) are algebraically independent polynomials in \(k[X_1, \dots, X_n]\). 
Another important application is to systems of invariant polynomials under finite reflection groups, 
since by the Chevalley-Shephard-Todd theorem their invariant rings are polynomial algebras. 
Typical examples include the symmetric groups \(S_n\), the hyperoctahedral groups \(B_n\), 
the dihedral groups \(I_2(m)\), and the exceptional finite reflection groups \(E_6, E_7, E_8, F_4, H_3, H_4\).
\end{abstract}




\maketitle

\section{Introduction}

Solving systems of polynomial equations is a central problem in computer algebra and symbolic computation, with applications ranging from algebraic geometry and invariant theory to optimization and real algebraic geometry.
Given a system
\[
\f = (f_1,\dots,f_n) \subset k[X_1,\dots,X_n],
\] where \(k\) is a field, 
solving $\f = 0$ can be interpreted in different ways. From a numerical
point of view, it means approximating the isolated points of the variety
defined by the ideal generated by $\f$ while from a symbolic computation
 perspective, it means providing an exact description, for instance through zero-dimensional parametrizations such as rational univariate representations or geometric resolutions, 
of the solution set that allows one to extract further algebraic or
geometric information efficiently.  
In general, polynomial system solving, which is known to be NP-hard \cite{fraenkel1979complexity, garey2002computers}, is a difficult problem. In this
paper, we focus on the computer algebra setting and consider
systems with additional algebraic structure, which allows for more
efficient algorithms.

\subsubsection*{Our problem.} In many applications, polynomial systems exhibit additional algebraic structure that can be exploited to design faster algorithms. In this work we focus on \emph{composable polynomial systems}, where each equation can be written as
\[
f_i(\X)=h_i\big(g_1(\X),\dots,g_n(\X)\big), \qquad i=1,\dots,n,
\]
for polynomial maps $\g=(g_1,\dots,g_n)$ and $\h=(h_1,\dots,h_n)$.

By exploiting the compositional structure, instead of solving $\f(\X)=0$ directly, whose complexity depends on the degrees of the composed polynomials $f_i$, one may first solve the \emph{outer system}
\(
\h(\Y)=0
\)
in new variables $\Y$, and then lift its solutions through the \emph{inner map} $\g$.
Since the degrees of $\h$ and $\g$ are often much smaller than those of $\f$, this approach can significantly reduce the computational cost.


\subsubsection*{Applications.}
An important special case occurs when $f_1, \dots, f_n$ belong to
the subring $k[g_1, \dots, g_n]$, where $g_1, \dots, g_n$ are algebraically
independent. In this situation, each $f_i$ admits a unique representation
$f_i = h_i(g_1, \dots, g_n)$ for some polynomial
$h_i \in k[Y_1, \dots, Y_n]$, and the composable framework applies directly.

Another natural class of examples arises in invariant theory. Let
$G \subset \mathrm{GL}_n(k)$ be a finite reflection group. By the
Chevalley--Shephard--Todd theorem \cite{chevalley1955invariants, sturmfels2002solving}, the invariant ring
$k[X_1,\dots,X_n]^G$ is a polynomial algebra of the form
$k[g_1, \dots, g_n]$. Consequently, every $G$-invariant polynomial $f_i$
can be written as a composition $f_i = h_i(g_1, \dots, g_n)$.
Typical examples include the symmetric groups $S_n$, the hyperoctahedral
groups $B_n$, the dihedral groups $I_2(m)$, and the exceptional reflection
groups $E_6, E_7, E_8, F_4, H_3,$ and $H_4$. Exploiting this structure can reduce the number of points to compute and the overall arithmetic cost of solving $\f = 0$.

This approach has been studied in various contexts. Classical invariant-theoretic algorithms are discussed in \cite{sturmfels1993algorithms, derksen2002computational}, while applications to reflection groups appear in \cite{lehrer2009unitary}. It has also been used to develop fast algorithms for solving invariant polynomial systems, including the connectivity of semi-algebraic sets \cite{riener2024connectivity, riener2025deciding}, emptiness of algebraic sets \cite{labahn2023faster}, computing critical points of polynomial maps \cite{vu2022computing, faugere2023computing}, polynomial optimization \cite{riener2013exploiting}, and the Euler-Poincar\'e characteristic \cite{basu2017efficient}.

In~\cite{vu2025computing}, we developed a randomized algorithm that,
given polynomial maps $\g = (g_1,\dots,g_n)$ and $\f = (f_1,\dots,f_n)$,
computes polynomials $\h = (h_1,\dots,h_n) \in k[Y_1,\dots,Y_n]$ such that
$\f = \h \circ \g$, that is, each $f_i$ is expressed as a polynomial in
$g_1,\dots,g_n$. The arithmetic complexity is softly polynomial in the
sizes of the straight-line programs computing $\f$ and $\g$, as well as in
the degree bound of the representation. In the important case where $\f$
is invariant under a finite reflection group whose invariant ring is
generated by $g_1,\dots,g_n$, the degrees of the $h_i$ are bounded by
$\max_i \deg(f_i)$, leading to particularly efficient computations.

Therefore, for these structured classes of systems, one can first apply
the representation algorithm of~\cite{vu2025computing} to compute the outer
map $\h$ from $\f$ and $\g$, and then apply the compositional solving
algorithm developed in this paper to the pair $(\g,\h)$. This yields an
efficient procedure for computing the solutions of the original system
$\f = 0$.


\subsubsection*{Contribution}

We present a probabilistic symbolic algorithm for computing all isolated regular solutions of composable systems $\f=\h\circ\g$.
Our method combines geometric resolutions with a global Newton-Hensel (homotopy) lifting procedure.
First, we compute a zero-dimensional parametrization of the regular solutions of the outer system $\h=0$ using symbolic homotopy techniques.
Then, we lift these solutions through the inner map $\g$ by a parametric Newton-Hensel scheme, and finally discard points where the Jacobian of $\g$ is singular.
All computations are performed using straight-line programs, ensuring complexity bounds polynomial in the input size.

The resulting algorithm computes a geometric resolution of the isolated regular solutions of $\f=0$ with arithmetic complexity
softly polynomial in the straight-line program sizes of $\h$ and $\g$ and in the quantities
\[
C=\prod_i \deg(h_i), \qquad
D=\prod_i \deg(g_i),
\]
which bound respectively the number of solutions of the outer and inner systems.
In contrast, solving the composed system directly would typically depend on the much larger Bézout bound $\prod_i \deg(f_i)$.
Thus, our complexity depends separately on the inner and outer maps rather than on their composition, yielding substantial savings for structured systems.

While homotopy lifting and geometric resolutions are classical tools, their combination with compositional structure has not been previously exploited to obtain complexity bounds depending separately on the inner and outer maps. The main novelty of this work lies in this structured decomposition and the resulting reduction from Bézout-type bounds to \(CD\).

\begin{theorem}
Let $k$ be a field of characteristic zero or sufficiently large positive characteristic. Let 
$\f = (f_1, \dots, f_n)$ be polynomials in $k[X_1, \dots, X_n]$, and assume that each $f_i$ is composable as $f_i = h_i(g_1, \dots, g_n)$ for $i = 1, \dots, n$. 

Then there exists a randomized algorithm 
\({\sf Solve\_h\_circ\_g}\) that takes $\h = (h_1, \dots, h_n)$ and $\g = (g_1, \dots, g_n)$ as input and outputs 
a zero-dimensional parametrization of the isolated regular points of $\f = 0$. 
The complexity of this algorithm is
\[
\softO\!\big( n (L_\h + L_\g + n^2(\gamma + \sigma)) (C E + D J + D C) \big)
\] 
operations in $k$, where \newline
- $L_\h$ and $L_\g$ are respectively the lengths of the straight-line programs computing $\h$ and $\g$,\newline
- $C = \deg(h_1) \cdots \deg(h_n)$, $E = (\deg(h_1)+1) \cdots (\deg(h_n)+1)$, \newline 
- $D = \deg(g_1) \cdots \deg(g_n)$, $J = (\deg(g_1)+1) \cdots (\deg(g_n)+1)$, \newline
- $\gamma = \max_{1 \le i \le n} \deg(h_i)$, and $\sigma = \max_{1 \le i \le n} \deg(g_i)$. 
\end{theorem}
In this paper, by a randomized (or probabilistic) algorithm, we mean an algorithm that performs correct computations by making random choices of points lying in a certain non-empty Zariski open subset of appropriate affine spaces.
In this sense, the algorithm is of Monte Carlo type, i.e., it returns the correct output with high probability, at least some fixed constant greater than \(1/2\). By repeating the algorithm multiple times, the error probability can be made arbitrarily small; this can be quantified using the Schwartz-Zippel lemma (see e.g. \cite{schwartz1980fast} or \cite[Lemma 6.44]{vonzurGathenGerhard2013}).


\subsubsection*{Structure of the paper}
The paper is organized as follows.
Section \ref{sec:pri} introduces the algebraic and algorithmic preliminaries used throughout the paper, including complexity notations, straight-line program representations, and zero-dimensional parametrizations such as rational univariate representations and geometric resolutions.
Section \ref{sec_glob} reviews the global Newton-Hensel lifting technique, which serves as the main symbolic continuation tool for refining parametrizations along a homotopy.
Section \ref{sec:core} presents the core contributions of the paper: we first develop a symbolic homotopy algorithm for computing the regular solutions of a square polynomial system in the outer variables, and then exploit the compositional structure $\f=\h\circ\g$ to lift these solutions through the inner map.
These ingredients are combined into the complete solving procedure together with a detailed complexity analysis, yielding bounds that depend separately on the degrees and straight-line program sizes of $\h$ and $\g$.
Examples illustrate how exploiting this structure leads to significant improvements compared to solving the system  \(\f\) directly.


\section{Preliminaries}
\label{sec:pri}
\subsection{Complexity Notations} The following notation will be used to measure the complexities of our algorithms.  The notation \(f \in \softO(g)\) means that there exists a constant \(a\) such that \(f\) is in \(O(g \log(g)^a)\).  

We denote by $\omega$ the matrix multiplication exponent. Currently, $\omega < 2.373$. In complexity estimates, terms of the form $O(n^\omega)$ account for the cost of  linear algebra operations such as matrix inversion or solving linear systems.

For positive integers \(d\), let ${\sf M}(d)$ denote the cost of multiplying univariate polynomials of degree at most $d$ over the base ring $\K$. Using the algorithms of Schönhage and Strassen~\cite{schonhage1971schnelle}, Schönhage~\cite{schonhage1977schnelle}, and Cantor and Kaltofen~\cite{cantor1991fast}, this can be done in $\softO(d) = O(d \log d \log \log d)$ operations. Over a finite field with $q$ elements, assuming the existence of a Linnik constant, one can take ${\sf M}(d) = O(d \log q \log (d \log q))$ using the result in~\cite{harvey2022polynomial}.


\subsection{Data Representations}
Let $h \in \K[x_1, \dots, x_n]$ be a polynomial of degree $d$. The usual dense representation encodes all $\binom{n+d}{n}$ coefficients, while a sparse representation stores only the nonzero terms. In this work, we use a {\em straight-line program (SLP)}, also known as an {\em algebraic circuit}, as an alternative representation.

An SLP computing polynomials $h_1, \dots, h_m \in \K[x_1, \dots, x_n]$ is a sequence
\[
\gamma = (\gamma_{-n+1}, \dots, \gamma_0, \gamma_1, \dots, \gamma_L),
\]
where $\gamma_{-n+1} := x_1, \dots, \gamma_0 := x_n$, all $h_i$ are among the $\gamma_k$, and for $k > 0$, each $\gamma_k$ is of the form
\[
\gamma_k = a * \gamma_i \quad \text{or} \quad \gamma_k = \gamma_i * \gamma_j,
\]
with $i, j < k$, $a \in \K$, and $* \in \{+, -, \times\}$. The integer $L$ is the SLP’s length. Notably, polynomials with many monomials can admit very short SLPs; for example, $(x+1)^k$ has $k+1$ terms but can be computed with a program of length $O(\log k)$.

SLPs were first studied in probabilistic polynomial identity testing and later used in computer algebra for univariate elimination and factorization problems~\cite{heintz1981absolute, kaltofen1988greatest, kaltofen1989factorization}, and polynomial system-solving algorithms based on Newton-Hensel lifting~ (see e.g.,  \cite{giusti1997lower, giusti1998straight, giusti1995polynomial, GiustiLecerfSalvy2001, Schost03} and references therein).

This representation is general: any polynomial of degree $d$ in $n$ variables can be computed by an SLP of length at most $3 \binom{n+d}{n}$, and a sparse polynomial with $N$ nonzero terms can be computed with $L = O(Nd)$. Moreover, if two $n$-variate power series are given at precision $d$ by SLPs of length $L$, their product truncated at degree $d$ can be computed by an SLP of length $O(d^2 L)$. In fact, all homogeneous components up to degree $d$ can be computed using an SLP of length $(d+1)^2 (L+1)$~\cite{krick1996computational}.

Let $L_\h$ and $L_\g$ denote the lengths of SLPs for polynomials $\h = (h_1, \dots, h_m) $ in  $ k[Y_1, \dots, Y_n]$ and for polynomials $\g = (g_1, \dots, g_n)$ in  $k[X_1, \dots, X_n]$, respectively.  
By appending the SLP of $\h$ after computing $\g$ and substituting $Y_i \mapsto g_i(X_1, \dots, X_n)$, we obtain a SLP for $\f = \h(\g)$.  
In particular, the length of a SLP for $\f = \h(\g)$ is at most $L_\h + L_\g$.


\subsection{Zero-dimensional Parametrizations  }

Let $V \subset \kbar{}^n$ be a finite set defined by polynomials in
$k[X_1,\dots,X_n]$. A \emph{rational univariate representation} (RUR)
$\mathscr{R}=((q,v_1,\dots,v_n),\lambda)$ of $V$ consists of:
\begin{itemize}
\item a square-free polynomial $q \in k[T]$, where $T$ is a new indeterminate
      and $\deg(q)=|V|$;

\item \(n\) univariate polynomials $v_1,\dots,v_n \in k[T]$, each of degree at most $\deg(q)$,
      such that
      \[
      q(T)=0, \
      X_i=\frac{v_i(T)}{q'(T)} \ (i=1,\dots,n), \text{\ where \ } q'=\frac{\partial q}{\partial T};
      \]
\item a linear form $\lambda=\lambda_1X_1+\cdots+\lambda_nX_n$ with coefficients in $k$ such that
      \[
      \lambda(v_1,\dots,v_n) \equiv Tq' \pmod{q},
      \]
      i.e., the roots of $q$ are precisely the values taken by $\lambda$ on $V$.
\end{itemize}
Then we write \(V= Z(\scrR)\). This definition implies that $\lambda$ takes pairwise distinct values on $V$; it is
therefore called a \emph{separating linear form}. One advantage of RUR is that the rational parametrization, with $q'$ in the
denominator, provides effective control of coefficient growth when $k=\mathbb{Q}$,
or of the degree in a parameter when working over a rational function field
$k(y)$ (cf.\ \cite{alonso1996zeros, GiustiLecerfSalvy2001, rouillier1999solving}).

An alternative representation is the \emph{geometric resolution}. A geometric
resolution $\scrS = ((q,w_1,\dots,w_n),\lambda)$ of $V$ consists of a separating linear form
$\lambda$ and polynomials $q,w_1,\dots,w_n$ in $k[T]$ such that $q$ is monic and
square-free, $\deg(w_i)<\deg(q)$ for all $i$, and
\[
q(T)=0, \ 
X_i = w_i(T) \ (i=1,\dots,n), \
\lambda(w_1,\dots,w_n)=T.
\]

If $\mathscr{R}=((q,v_1,\dots,v_n), \lambda)$ is a rational univariate representation
of $V$, then, since $q$ is square-free, $q'$ is invertible modulo $q$.
Computing the inverse of $q'$ modulo $q$ yields a geometric resolution
$((q,w_1,\dots,w_n),u)$ with
\[
w_i \equiv (q')^{-1} v_i \pmod{q}.
\]
We denote by {\sf RUR\_to\_GR} the procedure that performs this conversion. Let \(d = \deg(q) = \deg(V)\). 
Using fast polynomial multiplication, the inverse can be computed in
$O({\sf M}(d)\log d)$ operations in $k$, and the $n$ modular multiplications cost
$O(n{\sf M}(d))$. Hence, the total cost of the conversion is
$\softO(n d)$ arithmetic operations in $k$. In particular, this step is
quasi-linear in the output size and negligible compared to the cost of
computing the parametrization itself. 
In what follows, we freely switch between these two equivalent representations depending on which form is more convenient for complexity analysis.


\section{Global Newton-Hensel Lifting}
\label{sec_glob}
Let $\f = (f_1(\X, T), \dots, f_n(\X, T))$, with $\X = (X_1, \dots, X_n)$. 
Starting from a parametrization $X_i = v_i(U)$ and a minimal polynomial $q(U)$
that describe the solutions of $\f(\X, T) = 0$ modulo $T^m$, the goal is to compute refined polynomials
$V_i(U, T)$ and $Q(U, T)$ such that
\[
\f(\V) \equiv 0 \pmod{\langle T^\delta, Q \rangle}, \qquad \deg_T(V_i) < \deg_T(Q) \le \delta.
\]
This procedure is called the {\sf GLS\_Lifting} algorithm and corresponds to a particular case of \cite[Algorithm 1]{GiustiLecerfSalvy2001}.

Under the assumptions that $\f(\X, T) = 0$ modulo $T^m$ and that the Jacobian $J_{\f, \X}$ 
 at \((v_1(U), \dots, v_n(U))\) is invertible, the procedure {\sf GLS\_Lifting} in Algorithm \ref{alg:GLS} is a Newton-Hensel lifting scheme. Each iteration doubles
the $T$-adic precision while preserving the primitive element relation \(\lambda(v_1(U), \dots, v_n(U)) = U\).
Consequently, after $O(\log(\delta))$ iterations the algorithm returns a parametrization
$(Q(T, U), \V(T, U))$ satisfying $\f(\V) \equiv 0 \pmod{\langle T^\delta, Q \rangle}$. 

\begin{theorem}[\cite{GiustiLecerfSalvy2001}]
Algorithm \ref{alg:GLS} is correct and its complexity is
\[
O((nL + n^\omega) \, {\sf M}(\deg(q)) \, {\sf M}(\delta)) \subset \softO((nL + n^\omega) \, \deg(q) \, \delta)
\]
arithmetic operations in $k$, where $L$ is the length of a straight-line program for $\f$. 
\end{theorem}

\begin{algorithm}[h] 	 
\caption{GLS\_Lifting($\f, q, \v, \lambda, \delta$)}
\label{alg:GLS}

\begin{flushleft}
{\bf \texttt{Input:}}
\begin{itemize}
  \item polynomials $\f = (f_1, \dots, f_n)$ in $k[X_1, \dots, X_n, T]$,
  \item a monic polynomial $q(U) \in k[U]$,
  \item polynomials $\v = (v_1(U), \dots, v_n(U))$ with $\deg(v_i) < \deg(q)$, 
  \item a linear form $\lambda = \lambda_1 X_1 + \dots + \lambda_n X_n$.
\end{itemize}

{\bf \texttt{Assumptions:}}
\begin{itemize}
  \item $\f(\v) \equiv 0 \pmod{\langle T^m, q(U) \rangle}$,
  \item the Jacobian of \(\f\) with respect to \(\X\) is invertible at \(\v\),
  \item $\lambda(\v) \equiv U \pmod{q(U)}$.
\end{itemize}

{\bf \texttt{Output:}}  $(Q, \V) \in k[T, U]^{n+1}$ such that 
$\f(\V) \equiv 0 \pmod{\langle T^\delta, Q \rangle}$
\end{flushleft}

\begin{enumerate}
\item[{\small 1:}] $\V \gets \v$, $Q \gets q$, $k \gets m$

\item[{\small 2:}] $J \gets \frac{\partial \f}{\partial \X}$

\item[{\small 3:}] \textbf{while} \(k < \delta\) \textbf{do}:
    \item[{\small 4:}] \quad $\V \gets \V - J(\V)^{-1} \f(\V) \mod \langle T^k, Q \rangle$
    \item[{\small 5:}] \quad $\Delta \gets \lambda(\V) - U$
    \item[{\small 6:}] \quad $\V \gets \V - (\frac{\partial \V}{\partial U} \cdot \Delta \mod \langle T^k, Q \rangle)$
    \item[{\small 7:}] \quad $Q \gets Q - (\frac{\partial Q}{\partial U} \cdot \Delta \mod \langle T^k, Q \rangle)$
    \item[{\small 8:}] \quad $k \gets 2k$ 
\item[{\small 9:}] \textbf{end while}
\item[{\small 10:}] \textbf{return} $(Q, \V)$
\end{enumerate}
\end{algorithm}

\begin{example}\label{ex:1}
Consider polynomials $f_1(X_1,X_2,T) = X_1 + X_2 - T -1$, $f_2(X_1,X_2,T) = X_1 X_2 - T$, and 
\[
q(U) = U^2 - 4U + 3, \quad 
v_1 = \frac{3}{2} - \frac{U}{2}, \quad 
v_2 = -\frac{1}{2} + \frac{U}{2}.
\]
Then we have $f_1(v_1,v_2) \equiv -T \pmod{\langle q(U) \rangle}$ and
$f_2(v_1,v_2) \equiv -T \pmod{\langle q(U) \rangle}$.
With $\delta = 2$ and $\lambda = X_1 + 3 X_2$, Algorithm \ref{alg:GLS} returns
\[
Q(T, U) = U^2 -4U + 3 + 2T(5-2U),
\] and \[
V_1 = \frac{3}{2} - \frac{U}{2} + \frac{3T}{2}, \quad
V_2 = -\frac{1}{2} + \frac{U}{2} - \frac{T}{2}.
\]
\end{example}



\section{Solving Systems of Composable Polynomials}
\label{sec:core}
Let $\g = (g_1,\dots,g_n) : k^n \to k^n$ and 
$\h = (h_1,\dots,h_n) : k^n \to k^n$ be polynomial maps.
We consider systems of the form
\[
\f(\X) = \h(\g(\X)) = 0,
\]
where $\X = (X_1, \dots, X_n)$. Then solving $\f(\X) = 0$ is equivalent to solving the coupled system
\begin{equation} \label{eq:decomp}
\begin{cases}
    h_1(Y_1, \dots, Y_n) &= 0, \\
    \hspace{0.5cm}  \vdots & \\
    h_n(Y_1, \dots, Y_n) &= 0,
\end{cases}
\ \text{and} \
\begin{cases}
g_1(X_1, \dots, X_n) - Y_1 &= 0, \\
\hspace{0.5cm}  \vdots & \\
g_n(X_1, \dots, X_n) - Y_n &= 0.
\end{cases}    
\end{equation}

In other words, we first solve the outer system $\h(\Y) = 0$ for $\Y = (y_1,\dots,y_n)$ and then lift each solution through the inner map $\g$ to obtain solutions of $\f$. This compositional structure will allow us to reduce the computational cost compared to solving $\f$ directly, as the complexity depends on the degrees of $\h$ and $\g$ individually rather than the (potentially much larger) degree of $\f$.

\subsection{Solving square polynomial systems}
\label{subsec_square1}

We first consider the problem of computing the regular points of a square system
\[
h_1(\Y) = \cdots = h_n(\Y) = 0, \quad \text{where } \Y = (Y_1, \dots, Y_n),
\]
consisting of $n$ equations in $n$ variables, using symbolic homotopy techniques.  
Solving such square systems is a classical problem in symbolic computation, and efficient algorithms are known.  
However, to our knowledge, there is no single reference that provides a complete, detailed algorithm. For this reason, this subsection is devoted to presenting a full procedure.

We also refer the reader to \cite{el2018bit, vu2020homotopy, hauenstein2021solving}, and references therein, for related discussions and problems.  
Algorithm \ref{al:homo_solve} below outputs the \emph{regular points} of $\h = 0$, i.e., the points where the Jacobian of $\h$ has full rank.  
Note that isolated points can also be obtained by a minor modification of the last step of this algorithm by using \cite[Proposition 12]{hauenstein2021solving}. 

The {\sf Homotopy\_Nonsingular} algorithm solves a square polynomial system by first constructing a simple ``start system'' whose solutions are easily computed as intersections of random linear forms. It then continuously deforms this system into the target system using a homotopy and lifts the known solutions along this deformation using a Newton–Hensel style lifting procedure ({\sf GLS\_Lifting}). Rational function reconstruction finally recovers exact solutions over a rational function field, from which we recover a rational univariate representation over the base field. At the final step, we remove the singular points.

\begin{lemma} \label{lemma:remove}
Let \(\h = (h_1, \dots, h_n)\) be a polynomial system in \(k[\Y]\) with \(\Y = (Y_1, \dots, Y_n)\). 
Let \((q, (v_1, \dots, v_n), \lambda)\) be a RUR of a finite set \(W \subset V(\h) \cap \kbar{}^n\). 

Then there exists a procedure {\sf Remove} that removes all the singular points from \(W\) using 
\(\softO(C^2(n^3 + n \, L_\h))\) operations in \(k\), where \(L_\h\) is the length of a straight-line program computing \(\h\) and \(C = \deg(q)\).
\end{lemma}
\begin{proof}
We remove all the points at which the Jacobian matrix \(J_\h\) of \(\h\) vanishes.  
To do this, we first construct a straight-line program of length \(O(nL_\h)\) computing \(J_\h\).  
Next, we evaluate this matrix modulo \(q\) and use Gaussian elimination modulo \(q\) to identify the divisors of \(q\) that need to be removed.  
As explained in \cite[Section 5]{hauenstein2021solving}, this can be done using 
\(\softO(C^2(n^3 + nL_\h))\) operations in \(k\).  
The final cleaning step is performed using the {\sf Clean} algorithm from \cite[Algorithm 10]{GiustiLecerfSalvy2001}, whose cost is dominated by the previous bound.
\end{proof}

\begin{algorithm}[h] 	 
\caption{\sf Homotopy\_Nonsingular$(\h)$}
\label{al:homo_solve}

\begin{flushleft}
{\bf \texttt{Input:}} $\h = (h_1(\Y), \dots, h_n(\Y))$ with $\Y = (Y_1, \dots, $ $Y_n)$

{\bf \texttt{Output:}}  a RUR of regular points of $\h=0$
$\f(\V) \equiv 0 \pmod{\langle T^\delta, Q \rangle}$
\end{flushleft}

\begin{enumerate}
\item[{\small 1:}] \textbf{for} {$i=1, \dots, n$} \textbf{do}: \label{step:1}

\item[{\small 2:}] \quad construct a ``start system'' $p_i$ as a product of $\deg(h_i)$
          random linear forms 
          \[p_i(\Y) = \prod_{j=1}^{\deg(h_i)} \Big(\lambda_{i,j,0} + \lambda_{i,j,1} Y_1 + \cdots + \lambda_{i,j,n} Y_n \Big)\]

\item[{\small 3:}]  \textbf{end do}:
\item[{\small 4:}] compute a RUR $\scrR_0 = (q_0, (v_{0,1}, \dots, v_{0,n}), u)$ of the start system $\p = (p_1, \dots, p_n)$

\hfill{\(\triangleright\) \ {\(q_0, v_{0, i} \in k[U]\)}}

\item[{\small 5:}] define a homotopy $\r(\Y,T) = (1-T)\cdot \p(\Y) + T\cdot \h(\Y)$ that  deforms $\p$ into $\h$

\item[{\small 6:}] $\scrR \gets$ {\sf GLS\_Lifting}$(\r, q_0, \v_0, u, E)$ \label{step:scrE}

\hfill{\(\triangleright\) {lift $\scrR_0$  along the homotopy  to a RUR $\scrR$ with coefficients in $k[[T]]/\langle T^{2E}\rangle$}}

\hfill{\(\triangleright\) {\(E\): degree of the homotopy curve \(\r\)}}

\item[{\small 7:}] perform rational function reconstruction to obtain a RUR $\scrS$ over $k(T)$

\item[{\small 8:}] $\scrR_1 \gets $ a RUR with coefficients in \(k\) obtained from \(\scrS\)
    \label{step:8}

\item[{\small 9:}] remove from $Z(\scrR_1)$ all points that are singular in $V(\h)$ \label{step:fin}
\end{enumerate}
\end{algorithm}

\begin{theorem}\label{thm:first}
Let $\h=(h_1,\dots,h_n)\in k[\Y]^n$ be a square system and let $L$ be the length of a straight-line program computing $\h$.  

There exists a randomized algorithm  {\sf Homotopy\_Nonsingular} that  computes a rational univariate representation of all regular solutions of $\h=0$ using
\[
\softO\!\big( CEn(L_\h+n^2\gamma) \big) \ \text{operations in } k,
\]
where $C=\deg(h_1)\cdots\deg(h_n)$, \(\gamma = \max_{1\le i \le n}(\deg(h_i))\), and
$E=(\deg(h_1)+1)\cdots(\deg(h_n)+1)$.
\end{theorem}
\begin{proof}
For $i = 1, \dots, n$, the polynomial $p_i$ is a product of $\deg(h_i)$ linear forms $\ell_{i,j}$.  
Each linear form can be computed in $O(n)$ operations in $k$, which implies that 
$\p = (p_1, \dots, p_n)$ can be computed by a straight-line program of length 
\(
O\big(n (c_1 + \cdots + c_n)\big)
\) 
operations in $k$, where $c_i = \deg(h_i)$.  

Moreover, the solutions of $\p = 0$ are obtained by setting one factor of each $p_i$ to zero.  
Since each $p_i$ is a product of random linear forms and $k$ is a field of characteristic zero, all linear forms 
\(
(\ell_{i,j})_{i=1,\dots,n,\, j=1,\dots,c_i}
\)
are pairwise distinct. Therefore, the system $\p = 0$ has 
\(
C = c_1 \cdots c_n
\) 
solutions. Each corresponding linear system can be solved in $O(n^3)$ operations, e.g., by Gaussian elimination.  
Hence, computing all solutions of $\p = 0$ requires 
\(
O(C n^3)
\) 
operations in $k$.  

Knowing all the points of $\p = 0$, we can construct a rational univariate representation $\scrR_0$ 
such that $Z(\scrR_0) = V(\p)$ in time $\softO(C n)$ using fast interpolation \cite[Chapter 10]{vonzurGathenGerhard2013}.  
Finally, we remark that, locally around any solution of $\p = 0$, the system $\p = 0$ is equivalent to a linear system.  
Consequently, the Jacobian matrix of $\p$ is invertible at all its roots.

The system $\r \in k[\Y, T]$ is defined as 
\(
\r(\Y, T) = (1-T) \cdot \p(\Y) + T \cdot \h(\Y),
\) 
so that in particular $\r(\Y, 0) = \p$ and $\r(\Y, 1) = \h$.  
The lengths of the straight-line programs to compute $\h$ and $\p$ are $L$ and 
\(
O\big(n (c_1 + \cdots + c_n)\big),
\) 
respectively (as discussed above). Therefore, the system $\r$ can be computed in 
\(
O(L'), \text{ with } L' = L_\h + n (c_1 + \cdots + c_n) = L_\h + n^2 \gamma,
\) 
operations in $k$, with \(\gamma = \max(c_1, \dots, c_n)\).  

Moreover, from the previous discussion, the Jacobian matrix of $\p = \r(\Y, 0)$ with respect to $\Y$ has full rank at all its solutions.  
Hence, all the conditions required to perform the algorithm {\sf GLS\_Lifting} are satisfied.  

Let us write \(V(\r) = V(J) \cup V' \cup V''\). Here \(V(J)\) is the union of all one-dimensional irreducible components of \(V(\r) \subset \kbar{}^{n+1}\) such that projection onto the \(T\)-axis is dense, \(V'\) is the union of all other components of dimension one of \(V(\r)\), and \(V''\) is the union of the components of higher dimension. The zero-set \(V(J)\) is called the {\em homotopy curve}.  
Let \(\eta = \eta_0 + \eta_1 Y_1 + \cdots + \eta_n Y_n + \eta_{n+1} T\) be a generic hyperplane in \(T, Y_1, \dots, Y_n\).  
Then \(\big(V(J) \cup V'\big) \cap V(\eta)\) is a finite set consisting of \(\deg(V(J)) + \deg(V')\) points, while \(V'' \cap V(\eta)\) consists of components of positive dimension.  
Therefore, we take \(E\) to be the number of isolated points of \(V(\r) \cap V(\eta)\); equivalently, \(E = \deg(V(J))\), the degree of the homotopy curve. 

From \(\eta = 0\), we can rewrite \(T\) as 
\(\bar{\eta} = -(\eta_0 + \eta_1 Y_1 + \cdots + \eta_n Y_n)/\eta_{n+1}\).  
The points in \(V(\r) \cap V(\eta)\) are thus in one-to-one correspondence with the solutions of the system 
\(\bar \r = (1-\bar \eta)\cdot\h + \bar \eta\cdot\p\).  
The degrees of \(\bar \r\) are at most \((\deg(h_1)+1, \dots, \deg(h_n)+1)\).  
Thus, by Bézout's theorem, one can take 
\[
E \le \prod_{i=1}^n (\deg(h_i)+1).
\]

Since all the conditions to perform the algorithm {\sf GLS\_Lifting} are satisfied and the required precision in \(T\) is the degree \(E\) of the homotopy curve, computing the parametrization $\scrR$ at Step~\ref{step:scrE} requires
\(
O\big((n L' + n^\omega) \, C \, E\big)
\) 
operations in $k$.

As explained in \cite[Section 5]{hauenstein2021solving} and \cite[Section 2.2]{el2018bit}, we perform rational function reconstruction on all coefficients of $\scrR$, following the approach of \cite{Schost03}, to obtain a RUR $\scrS$ with coefficients in $k(T)$.  
Then, using \cite[Lemma 4.4]{rouillier2000finding}, we obtain a RUR $\scrR_1 = (q_1, (v_1, \dots, v_n), \lambda)$ with coefficients in $k$.  
Computing $\scrS$ requires $\softO(C E n)$ operations in $k$, and obtaining $\scrR_1$ from $\scrS$ requires the same complexity.

At the final stage, we remove all points that are not regular. 
This can be done using the {\sf Remove} algorithm from Lemma~\ref{lemma:remove}, 
which requires \(\softO(C^2(n^3 + nL))\) operations in \(k\). Thus the total compexity of Algorithm \ref{al:homo_solve} is \(\softO(CEn(L_\h + n^2\gamma)\) operations in \(k\). 
\end{proof}

Note  that when the polynomials $h_i$ have arbitrary supports and generic coefficients,
the system can be solved using the sparse symbolic homotopy algorithm of
\cite[Section~5]{jeronimo2009deformation}.
Moreover, when $h_1,\dots,h_n$ form a reduced regular sequence, one may instead
use the geometric resolution algorithm of Giusti et al.
\cite[Theorem~1]{GiustiLecerfSalvy2001}, which computes a geometric resolution
of $\h=0$ using
\[
\softO\!\big(n(nL_\h+n^{\omega}) d^{2}\delta\big)
\]
operations in $k$, where $\delta$ is the geometric degree of $\h$, bounded by $E$
in Theorem~\ref{thm:first}.

We conclude this subsection with a remark concerning the solution of the
original system $\f=0$. One may ask whether the algorithm {\sf Homotopy\_Nonsingular}
could be applied directly to $\f$. Although this is possible, the resulting
complexity would depend on the degrees $\deg(f_i)$ rather than on the smaller
degrees $\deg(h_i)$. In general, since $f_i \in k[g_1,\ldots,g_n]$, the degrees
$\deg(f_i)$ may grow as compositions of the $g_j$, and are typically much
larger than $\deg(h_i)$. Consequently, solving $\f=0$ directly would lead to a
significantly higher computational cost, which motivates the reduction to the
structured system $\h=0$.


\subsection{Solving the second system}
\label{subsec:2nd}

Let \(\left(q_h, ({\bar v}_1, \dots, {\bar v}_n), \mu\right)\) be a zero-dimensional parametrization of the nonsingular solutions of \(\h(Y_1, \dots, Y_n) = 0\). This can be computed using \({\sf Homotopy\_Nonsingular}(\h)\). Here, \({q_h}\) is a square-free polynomial in \(k[T]\) with \(\deg(q_h) \le C = \deg(h_1) \cdots \deg(h_n)\), the \({\bar v}_i\)'s belong to \(k[T]\) with \(\deg({\bar v}_i) < \deg(q_h)\) for all \(i = 1, \dots, n\), and \(\mu = \mu_1 Y_1 + \cdots + \mu_n Y_n\) is such that
\[
\mu({\bar v}_1, \dots, {\bar v}_n) \equiv T \frac{\partial q_h}{\partial T} \mod {q_h}(T).
\]

Setting 
\[
v_i = \frac{\partial q_h}{\partial T} \, {\bar v}_i(T) \mod {q_h}(T)
\]
gives a geometric resolution \((q_h, (v_1, \dots, v_n), \mu)\) of the isolated solutions of \(\h = 0\) with
\[
\mu(v_1, \dots, v_n) \equiv T \mod q_h(T).
\]
Then solving the system \(\g(\X) = 0\) is equivalent to solving
\[
q_h(T) = 0, \qquad
\begin{cases}
g_1(x_1, \dots, x_n) - v_1(T) = 0,\\
\hspace{0.5cm} \vdots\\
g_n(x_1, \dots, x_n) - v_n(T) = 0,
\end{cases}
\]
where \(q_h\) is square-free with \(\deg(q_h) \le C\), \(\deg(v_i) < \deg(q_h)\) for all \(i = 1, \dots, n\), and \(\mu(v_1, \dots, v_n) \equiv T \mod q(T)\).

To do it, we will study the system without the polynomial \(q_h(T)\), i.e., the system \(\g(\X) -  \v(T)\) of \(n\) polynomials in \(n+1\) variables \(\X, T\). We will find polynomials \(Q(U, T)\) and \(V_1(U, T), \dots, V_n(U,T)\) such that \[
Q(U, T) = 0, \qquad
\begin{cases}
g_1(V_1(U, T), \dots, V_n(U, T)) - v_1(T) = 0,\\
\hspace{0.5cm} \vdots\\
g_n(V_1(U, T), \dots,V_n(U, T)) - v_n(T) = 0. 
\end{cases}
\]
Since \(\deg_T(q_h) \le C \), it is sufficient to perform our computations to compute \(Q(U, T)\) and \(V_i(U, T)\) up to the precision \(T^C\).

\begin{algorithm}[h] 	 
\caption{{\sf Parametric}(\(\g, \v, C\))}
\begin{flushleft}
{\bf \texttt{Input:}} polynomials \(\g = (g_1, \dots, g_n)\) in \(k[X_1, \dots, X_n]\), a vector \(\v = (v_1, \dots, v_n) \in k[T]^n\), and a positive integer \(C\)

{\bf \texttt{Output:}} \(Q(U, T) \in k[U, T]\), \(\V = (V_1,\dots,V_n)\in k[U, T]^n\), and a linear form  \(\lambda = \lambda_1X_1 + \cdots + \lambda_nX_n\) such that\[\g(V_1(U, T), \dots, V_n(U, T)) - \v(T) = 0 \ \mod \ {\langle T^C, Q(U, T) \rangle}\]
\end{flushleft}

\begin{enumerate}
\item[{\small 1:}] \(s_0 \gets \) random point in \(k\)

\item[{\small 2:}] \(\bar{\F}(\X) \gets \g(\X) - \v(s_0)\)

\item[{\small 3:}]  \(\big(\qbar, (\wbar_1, \dots, \wbar_n), \lambda \big) \gets \) a geometric resolution of \(\bar{\F} = 0\) \label{step:shift}

\hfill{\(\triangleright\) {primitive element \(\lambda = \lambda_1X_1 + \cdots + \lambda_nX_n\)}}

\hfill{\(\triangleright\) {minimal polynomial \(\qbar(U) \in k[U]\)}}

\hfill{\(\triangleright\) {\(\bar{w}_i(U) \in k[U]\), \(\deg(\bar{w}_i) < \deg(\qbar)\), with \(X_i = \bar{w}_i(U)\)}}

\item[{\small 4:}]  \(S \gets T - s_0\)
\item[{\small 5:}]  \(\F(\X, S) \gets \g(\X) - \v(S + s_0)\)

\item[{\small 6:}] \((Q_{\text{shift}}, (V_{1,\text{shift}}, \dots, V_{n,\text{shift}})) \gets  \) \(  {\sf GLS\_Lifting}(\F, \qbar,  \bar{\w},  \lambda, 2C)\) \label{step:shift2}

\item[{\small 7:}] \(Q_{\text{poly}}(U, S) \gets \text{RationalReconstruction}(Q_{\text{shift}}, C)\)

\item[{\small 8:}] \textbf{for} \(i=1\) to \(n\) \textbf{do}

\item[{\small 9:}] \quad \(V_{i,\text{poly}}(U, S) \gets \text{RationalReconstruction}(V_{i,\text{shift}}, C)\)
\item[{\small 10:}] \textbf{end for}
\item[{\small 11:}] \(Q(U, T) \gets Q_{\text{poly}}(U, T - s_0)\)
\item[{\small 12:}] \textbf{for} \(i=1\) to \(n\) \textbf{do}
\item[{\small 13:}] \quad \(V_i(U, T) \gets V_{i,\text{poly}}(U, T - s_0)\)
\item[{\small 14:}] \textbf{end for}
\item[{\small 15:}] \textbf{return}  \((Q(U,T), (V_1(U,T),\dots,V_n(U,T)), \lambda)\)
\end{enumerate}
\end{algorithm}

\begin{theorem} \label{thm:DJC}
Let
\(\g = (g_1, \dots, g_n) \) be in \(k[X_1, \dots, X_n]\), 
\(\v = (v_1, \dots, v_n) \in k[T]^n\), and \(C\) be a positive integer. The algorithm 
\({\sf Parametric}(\g, \v, C)\) returns 
polynomials \(Q(T, U) \in k[U, T]\) and 
\(\V = (V_1, \dots, V_n) \in k[U, T]^n\) such that, for \(i = 1, \dots, n\),
\[g_i(V_1(U, T), \dots, V_n(U, T)) - v_i(T) = 0 \ \mod \ {\langle T^C, Q(U, T) \rangle}\]
Moreover, the complexity of this algorithm is 
\[
\softO\!\big( n (L_\g + n^2) \,  D \, (J+C) \big)
\] 
operations in \(k\), where \(D = \deg(g_1) \cdots \deg(g_n)\), 
\(J = (\deg(g_1)+1) \cdots (\deg(g_n)+1)\), and
\(L_\g\) is the length of a straight-line program computing \(\g\). 
\end{theorem}

\begin{proof}
For any \(\tau \in \overline{k}\), by the Bézout bound, the system 
\(\g(\X) - \v(\tau) = 0\) has at most 
\(
D = \deg(g_1) \cdots \deg(g_n)
\) 
isolated solutions.  
Let \(s_0 \in k\) be a point such that the fiber \(\g(\X) = \v(s_0)\) has exactly \(D\) isolated {regular} solutions and the linear form 
\(
\lambda = \lambda_1 X_1 + \cdots + \lambda_n X_n
\)
takes distinct values at distinct solutions of \(\g(\X) - \v(s_0) = 0\).  
Such a point \(s_0\) exists because:  
\begin{itemize}
    \item The set of \(T\) for which \(\v(T)\) is a {regular value} of \(\g\) is Zariski open (by the algebraic version of Sard's lemma; see e.g., \cite[Proposition B.3]{SafeySchost2017}).  
    \item The set of linear forms \(\lambda\) that separate all points in a fiber is Zariski open; hence a generic \(\lambda\) works for all fibers over a dense open subset of \(T\).
\end{itemize}
Let \(S = T - s_0\) and define 
\(
\F(\X, S) = \g(\X) - \v(S + s_0).
\) 
Then \(\F(\X, 0) = \g(\X) - \v(s_0)\).  
By construction, the geometric resolution 
\(
(\bar{q}(U), (\bar{w}_1(U), \dots, \bar{w}_n(U)), \lambda)
\) 
of \(\F(\X, 0) = 0\) satisfies
\[
\F(\bar{w}_1(U), \dots, \bar{w}_n(U), 0) \equiv 0 \mod \bar{q}(U),
\] 
and 
\[
\lambda(\bar{w}_1(U), \dots, \bar{w}_n(U)) \equiv U \mod \bar{q}(U).
\]  
Moreover, since all solutions of \(\F(\X, 0)\) are regular, the Jacobian matrix 
\(
J_{\X,\F}(\bar{\w}(U), 0) = J_{\g}(\bar{\w}(U))
\) 
is invertible modulo \(\bar{q}(U)\).  
Therefore, the input conditions for \({\sf GLS\_Lifting}\) are satisfied.

The \({\sf GLS\_Lifting}\) algorithm takes as input the system \(\F(\X, S)\) with parameter \(S\), 
the initial data 
\((\bar{q}(U), (\bar{w}_1(U), \dots, \bar{w}_n(U)))\) at \(S = 0\) with invertible Jacobian matrix, 
and a degree bound \(\delta\). 
It outputs 
\(
(Q_{\rm shift}, (V_{1,{\rm shift}}, \dots, V_{n,{\rm shift}}))
\) 
such that 
\[
\F(V_{1,{\rm shift}}, \dots, V_{n,{\rm shift}}, S) \equiv 0 \mod \langle S^\delta, Q_{\rm shift} \rangle,
\] 
and 
\[
\lambda \big(\V_{\rm shift}(U, S)\big) \equiv U \mod \langle S^\delta, Q_{\rm shift} \rangle.
\]  
Moreover, as above, we only need to do the computation at the precision satisfies \(\delta \ge 2C\). 
Therefore, it suffices to perform the \({\sf GLS\_Lifting}\) algorithm at precision \(\delta = 2C\).

Since the solution branches are algebraic of degree at most \(C\) in \(S\), 
and the power series converge to the algebraic functions \(\X(S)\) along each branch,  we need to  perform rational reconstruction at degree \(C\) to obtain polynomials 
\(
(Q_{\rm poly}, (V_{1,{\rm poly}}, \dots, V_{n,{\rm poly}}))
\) 
such that
\[
\F(V_{1,{\rm poly}}(U), \dots, V_{n,{\rm poly}}(U), S) \equiv 0 \mod \langle S^C, Q_{\rm poly}(U, S) \rangle,
\] 
and 
\[
\lambda \left(\V_{\rm poly}(U, S) \right) \equiv U \mod \langle S^C, Q_{\rm poly}(U, S)\rangle.
\]

Let \(Q(U, T) = Q_{\rm poly}(U, T- s_0)\) and \(V_i(U, T) = V_{i, \rm poly}(U, T- s_0)\). Since \(\F(\X, S) = \g(\X)- \v(S+ s_0)\) we have, for \(S = T- s_0\), 
\[\g(V_1(U, T), \dots, V_n(U, T)) - \v(T) \equiv 0 \ \mod \ \langle T^C, Q(U, T) \rangle \]
 and \[
 \lambda \left(V_1(U, T), \dots, V_n(U, T) \right)\equiv U \ \mod \ \langle T^C,  Q(U, T) \rangle.
 \]

Finally,  we show that the map
\[
\{(u, t) \in \kbar{}^2 : Q(u, t) = 0 \} \rightarrow Z, 
\quad (u, t) \mapsto (\V(u, t), t),
\]
where \(Z = V(\g(\X) - \v(T))\), is bijective at precision \(C\). 

Let \(({\bm \alpha}, t) \in Z\) and set \(u = \lambda({\bm \alpha})\).
Consider the branch through \(({\bm \alpha}, t)\). 
For generic \(s_0\), there is a unique branch of \(\g(\X) = \v(T)\) passing through \(({\bm \alpha}, t)\). 
This branch is computed by \({\sf GLS\_Lifting}\) starting from the fiber at \(s_0\). 
Since \(s_0\) is generic, the branch is defined at \(S = t - s_0\) and connects to \(S = 0\) by algebraic continuation. 
Hence \((u,t)\) lies on the lifted curve, so \(Q(u, t) = 0\). 
This proves that the above map is surjective.
For injectivity, suppose
\((\V(u, t), t) = (\V(u', t'), t')\).
Then \(t = t'\) and \(\V(u,t) = \V(u', t)\).
Applying the linear form \(\lambda\) yields
\(u = u' \mod Q(U,t)\).
Since \(Q(U,t)\) is square-free and \(\deg_U(Q) = D\), which equals the number of solutions in a generic fiber, we conclude that \(u = u'\).
Thus the map is injective.

\smallskip
We conclude the proof by providing the complexity analysis of the {\sf Parametric} algorithm.  
To compute a geometric resolution of \(\bar{\F}\) at Step~3, 
we first perform {\sf Homotopy\_Nonsingular} on the input system \(\bar{\F}\) to obtain a RUR of \(\bar{\F} = 0\), 
and then convert this RUR into a geometric resolution using {\sf RUR\_to\_GR}.  
By Theorem~\ref{thm:first}, this step requires 
\[
\softO\!\big( D J n (L_\g + n^2\sigma) \big)
\] 
operations in \(k\), where 
\(D = \deg(g_1) \cdots \deg(g_n)\),  
\(J = (\deg(g_1)+1) \cdots (\deg(g_n)+1)\), and \(\sigma = \max_{1 \le i \le n}(\deg(g_i)\).

At Step~6, the {\sf GLS\_Lifting} procedure requires
\[
\softO\big( (nL + n^\omega) \, \deg(\qbar) \, C \big) = \softO\big( n (L + n^2) \, D \, C \big)
\] 
operations in \(k\).  
Finally, performing rational reconstruction on all coefficients of 
\(Q_{\rm poly}(U, S)\) and \(V_{i,{\rm poly}}(U, S)\) requires 
\(\softO(C J n)\) operations in \(k\), 
which is dominated by the previous steps.  

Therefore, the total complexity of the algorithm is 
\[
\softO\!\big( n (L + n^2) \, D \, (J + C) \big)
\] 
operations in \(k\). 
\end{proof}

\begin{example} \label{ex:2}
Consider 
\(
g_1 = X_1 + X_2, \quad g_2 = X_1 X_2, \quad 
v_1(T) = T+1, \quad v_2(T) = T.
\) 
Then the system at \(T = 0\) is
\(
g_1(X_1, X_2) - v_1(0) = X_1 + X_2 - 1, \quad 
g_2(X_1, X_2) - v_2(0) = X_1 X_2.
\)

Take a random point \(s_0 = 0\). Performing 
\({\sf Homotopy\_Nonsingular}\) to obtain a RUR, followed by 
\({\sf RUR\_to\_GR}\), gives a geometric resolution of 
\(
g_1(X_1, X_2) - v_1(0) = g_2(X_1, X_2) - v_2(0) = 0
\) 
as
\[
\bar{q}(U) = U^2 - 4U + 3, \quad
\begin{cases}
\bar{w}_1(U) = \frac{3}{2} - \frac{U}{2},\\
\bar{w}_2(U) = -\frac{1}{2} + \frac{U}{2},
\end{cases}
\quad \text{and} \quad \lambda = X_1 + 3 X_2.
\]
Here, \(D = 2\) and \(C=2\).  The output of the \({\sf Parametric}\) algorithm is
\[
Q(U,T) = U^2 - (4T+4)U + (3T^2 + 10T + 3), \] and\[
w_1(U,T) = \frac{3}{2} - \frac{U}{2} + \frac{3T}{2}, \quad
w_2(U,T) = -\frac{1}{2} + \frac{U}{2} - \frac{T}{2}.
\]
\end{example}


\subsection{Merging two polynomial systems}

Let $(q_h(T), v_1(T), \dots, v_n(T), \mu)$ be a geometric resolution of the
zero set of $\h$. Then
\[
q_h(T)=0 \quad \text{and} \quad \h(v_1(T), \dots, v_n(T))=0,
\]
with $\deg(v_i) < \deg(q_h) \le C = \deg(h_1)\cdots\deg(h_n)$.

Next, we apply the algorithm ${\sf Parametric}(\g,\v,C)$ to obtain
polynomials
\[
(Q(U,T), V_1(U,T), \dots, V_n(U,T))
\]
parametrizing the solutions of $\g(\X)=\v(T)$.
Combining both parametrizations leads to the system
\begin{equation}\label{eq_elimi}
q_h(T)=0, \qquad
Q(U,T)=0, \qquad
X_i = V_i(U,T) \ (i=1,\dots,n).
\end{equation}
Therefore, solving $\f(\X)=\h(\g(\X))=0$ reduces to obtaining from \eqref{eq_elimi} a geometric resolution
\[
P(S)=0, \qquad
X_i = W_i(S) \ (i=1,\dots,n),
\]
which parametrizes the solutions of $\h(\g(\X))=0$.
To do this, we compute a geometric resolution for $(U,T)$ and then substitute the resulting parametrizations into $V_i(U,T)$ to obtain $W_i$.
\begin{theorem}
The Algorithm~\ref{alg:final} is correct, and its complexity is
\[
\softO\!\big( n (L_\h + L_\g + n^2 (\gamma + \sigma)) (C E + D J + D C) \big)
\] 
operations in $k$, where \newline
- $L_\h$ and $L_\g$ are the lengths of the straight-line programs computing $\h$ and $\g$, respectively, \newline
- $C = \deg(h_1) \cdots \deg(h_n)$, $E = (\deg(h_1)+1) \cdots (\deg(h_n)+1)$, \newline 
- $D = \deg(g_1) \cdots \deg(g_n)$, $J = (\deg(g_1)+1) \cdots (\deg(g_n)+1)$, \newline
- $\gamma = \max_{1 \le i \le n} \deg(h_i)$, and $\sigma = \max_{1 \le i \le n} \deg(g_i)$. 
\end{theorem}
\begin{proof}
Since $Q(U, T)$ is square-free in $U$ for each root $T$ of $q_h$, any $(\X, T)$ satisfying $\h(\g) = 0$ corresponds to a unique $(U, T)$ with 
\[
Q(U, T) = 0, \quad q_h(T) = 0, \quad \text{and} \quad \X = \V(U, T).
\]

Since $Q(U, T)$ is square-free in $U$ for each root $T$ of $q_h$, the linear form $\mu = U$ separates the solutions of the $(U,T)$ system.  Computing a geometric resolution of $Q = q_h = 0$ yields $(P(S), (\upsilon(S), \tau(S)), U)$. 

Moreover, the linear form $\lambda$ obtained in the {\sf Parametric} procedure satisfies 
\[
\lambda(V_1, \dots, V_n) \equiv U \pmod{Q(U, T)}.
\] 
Therefore, after substitution, we have 
\[
\lambda(W_1(S), \dots, W_n(S)) \equiv S \pmod{P(S)}.
\]
Hence, $\big(P(S), W_1(S), \dots, W_n(S)), \lambda\big)$ forms a geometric resolution of $\h(\g) = 0$. 

We now give the complexity of the algorithm. Computing a geometric resolution of $\h = 0$ requires 
\[
\softO\!\big( C E n (L_\h + n^2\gamma) \big)
\] 
operations in $k$ by Theorem~\ref{thm:first}, while the {\sf Parametric} procedure needs 
\[
\softO\!\big( n (L_\g + n^2\sigma) \, D \, (J + C) \big)
\] 
operations in $k$ by Theorem~\ref{thm:DJC}. At Step~\ref{Step:scrRR}, the system involves only $2$ variables, with 
\[
\deg_T(Q(U,T)) \le C,  \deg_U(Q(U,T)) \le D, \ \text{and} \ \deg(q_h(T)) \le C,
\]
so the cost of this step is negligible compared to the previous computations. Finally, the evaluations only require linear algebra operations. 

Thus, the total cost of the algorithm is
\begin{multline*}
  \softO\!\big( C E n (L_\h + n^2\gamma) + n (L_\g + n^2\sigma) \, D \, (J + C) \big) \\
 \subset  \softO\!\big( n (L_\h + L_\g + n^2(\gamma + \sigma)) (C E + D J + D C) \big)
\end{multline*}
operations in $k$.
 \end{proof}

\begin{algorithm}[h] 	 
\caption{$\mathsf{Solve\_h\_circ\_g}(\h, \g)$}
\label{alg:final}

\begin{flushleft}
{\bf \texttt{Input:}} polynomials $\h = (h_1,\dots,h_n) \in k[Y_1,\dots,Y_n]^n$ and polynomials $\g = (g_1,\dots,g_n) \in k[X_1,\dots,X_n]^n$ 

{\bf \texttt{Output:}}  a geometric resolution $(P(U), (W_1(U),\dots, W_n(U)), \lambda)$ of regular points of $\h(\g(\X)) = 0$. 
\end{flushleft}

\begin{enumerate}
\item[{\small 1:}]$(q_h(T), v_1(T),\dots,v_n(T), \mu) \gets$ geometric resolution of $V(\h)$ 

\hfill{\(\triangleright\) using {\sf Homotopy\_Nonsingluar} and then {\sf RUR\_to\_GR} }

\item[{\small 2:}] \(C \gets \deg(q_h)\)

\hfill{\(\triangleright\) \(C \le \deg(h_1) \cdots \deg(h_n)\)}
\item[{\small 3:}] \((Q(U, T), V_1(U, T), \dots, V_n(U, T), \lambda) \gets {\sf Parametric}(\g, \v, C)\)
\item[{\small 4:}]\(\scrR \gets\) a geometric resolution of \(Q(U, T) = q_h(t) = 0\) \label{Step:scrRR}

\hfill{\(\triangleright\) \(\scrR = \big(P_{\rm prec}(S), \upsilon(S), \tau(S), U\big)\)}

\item[{\small 5:}] \textbf{for} \(i=1, \dots, n\) \textbf{do}

\item[{\small 6:}] \quad \(W_{i, {\rm prec}}(S) \gets V_i(\upsilon(S), \tau(S))\)

\item[{\small 7:}] \textbf{end for}

\item[{\small 8:}]  \(\big(P(S), (W_1(S), \dots, W_n(S), \lambda)\big) \) 

\hfill{\(\gets\) \({\sf Remove}\big(P_{\rm prec}(S), (\W_{\rm prec}(S), \lambda)\big)\)}

\item[{\small 9:}] \textbf{return } \(\big(P(S), (W_1(S), \dots, W_n(S), \lambda)\big)\) 

\end{enumerate}
\end{algorithm}

\begin{example} \label{ex:final_ex}
Consider
\[
f_1 = X_1 + X_2 - X_1 X_2 - 1, \qquad
f_2 = X_1^2 X_2^2 + X_1 X_2 .
\]
These polynomials can be written as 
\(
f_1 = h_1(g_1, g_2)\) and \(f_2 = h_2(g_1, g_2),
\) 
where 
\[
\begin{cases}
h_1 = Y_1 - Y_2 - 1,\\
h_2 = Y_2^2 + Y_2,
\end{cases}
\qquad
\begin{cases}
g_1 = X_1 + X_2,\\
g_2 = X_1 X_2.
\end{cases}
\]
After applying {\sf Homotopy\_Nonsingular} to $(h_1, h_2)$ to obtain a RUR of the isolated points of $h_1 = h_2 = 0$, and then performing {\sf RUR\_to\_GR} on this RUR, we obtain a geometric resolution of regular points of $h_1 = h_2 = 0$:
\[
q_h(T) = T^2 + T, \quad
\begin{cases}
v_1(T) = T + 1,\\
v_2(T) = T.
\end{cases}
\]
As in Example~\ref{ex:2}, after performing {\sf Parametric} on $\g$, $\v$, and $C = 2$, we obtain the linear form \(\lambda = X_1 + 3X_2\) and 
\[
Q(U,T) = U^2 - (4T + 4)U + (3T^2 + 10T + 3), \
\begin{cases}
V_1(U,T) = \frac{3}{2} - \frac{U}{2} + \frac{3T}{2},\\
V_2(U,T) = -\frac{1}{2} + \frac{U}{2} - \frac{T}{2}.
\end{cases}
\] 
We compute a geometric resolution for \(Q(U,T) = q(T) = 0\) with
\[P(S) = S^4 - 4S^3 - S^2 + 16S - 12\] and \[\upsilon(S) = S, \tau(S) = \frac{4}{15} S^3 - \frac{9}{15} S^2 - \frac{16}{15} S - \frac{21}{15}. \]
Substituting $(U,T)=(\upsilon(S),\tau(S))$ into $V_1,V_2$ gives a geometric
resolution of the original system $f_1=f_2=0$:
\[
P(S) = S^4 - 4S^3 - S^2 + 16S - 12=0,
\] and 
\[
\begin{cases}
W_1(S)=V_1(\upsilon(S),\tau(S))
      = \frac{2}{5}S^3-\frac{9}{10}S^2-\frac{21}{10}S+\frac{18}{5},\\[4pt]
W_2(S)=V_2(\upsilon(S),\tau(S))
      = -\frac{2}{15}S^3+\frac{3}{10}S^2+\frac{31}{30}S-\frac{6}{5}.
\end{cases}
\]
Finally one can indeed verify that $P(S) = 0$ has four solutions $\{-2, 1, 2, 3\}$.   The corresponding solutions of the original system are 
\[
(X_1, X_2) = (W_1(S), W_2(S) \in \{(1,-1), (1,0), (-1,1), (0,1)\}.
\] 
\end{example}

We conclude this section with an important remark on the complexity of our algorithm. With the polynomials from Example \ref{ex:final_ex}, using the classical B\'ezout bound on the original system, the number of solutions is bounded by 
\(\deg(f_1) \cdot \deg(f_2) = 8\), whereas the bound on the number of points that our algorithm actually needs to compute is 
\(\deg(h_1) \cdot \deg(h_2) \cdot \deg(g_1) \cdot \deg(g_2) = 4\), 
which coincides exactly with the size of the solution set.

\section{Conclusions and Future Research}

We presented a probabilistic symbolic homotopy algorithm for computing the isolated regular solutions of composable polynomial systems \(\f = \h(\g)\).
By separating the outer system $\h$ from the inner map $\g$ and combining geometric resolutions with Newton-Hensel lifting, the proposed method achieves arithmetic complexity that scales with the individual degrees and straight-line program sizes of the components, rather than with the typically much larger degrees of the composed system $\f$.
This demonstrates that exploiting algebraic structure can substantially reduce the cost of symbolic polynomial system solving, particularly for systems arising from subring representations and invariant theory.

Furthermore, as discussed in Section \ref{subsec_square1}, a slight modification of the {\sf Remove} procedure allows us to compute all isolated solutions of the system $\f = 0$, and not only the regular ones. The resulting complexity remains polynomial in the size of the input and in the number of solutions, although with a higher polynomial degree due to the additional algebraic tests required to handle singular points.

Several directions remain for future work.
A first objective is to extend the approach beyond regular solutions and isolated solutions to handle  positive-dimensional components, for instance via deflation techniques or equidimensional decompositions.
Another direction is the development of sparse or multihomogeneous variants that better exploit monomial structure and supports.
It would also be valuable to investigate deterministic alternatives to some randomized steps. 
Moreover, applications to invariant and symmetric polynomial systems, polynomial optimization, and real solving suggest promising opportunities to integrate the compositional framework into broader symbolic-numeric pipelines.

For our algorithms to become practically applicable, efficient implementations are essential. Although we have already developed a first implementation of our algorithms to validate their correctness on several examples, substantial work remains to achieve high practical efficiency. The implementation of Straight-Line Programs is particularly subtle and has been the subject of significant recent work, notably~\cite{van2026towards}, as well as developments within the Mathemagix project. Developing efficient implementations and providing an open-source library constitute important directions for future work.

Finally, the proposed algorithm for solving composable polynomial systems opens several directions for applications in areas where structured nonlinear models arise naturally, such as chemical reaction networks, robotics and geometric control via nonlinear feature liftings of configuration space, and systems biology, where hierarchical interaction structures induce layered polynomial dynamics. In addition, multivariate public-key cryptosystems rely on the computational hardness of solving systems of composable multivariate polynomial equations over finite fields.


\begin{acks}
The author thanks the anonymous referees for their valuable comments and suggestions, which helped improve the quality of the paper.
\end{acks}

\balance
\bibliographystyle{ACM-Reference-Format}
\bibliography{sample-base}

@String{JACM = "J. ACM" }

@String{BIT = "{BIT}" }

@String{Computing = "Computing" }

@String{Computer = "{IEEE} Computer" }

@String{Springer = "Springer-Verlag" }

@article{schwartz1980fast,
  title={Fast probabilistic algorithms for verification of polynomial identities},
  author={Schwartz, Jacob T},
  journal={Journal of the ACM (JACM)},
  volume={27},
  number={4},
  pages={701--717},
  year={1980},
  publisher={ACM New York, NY, USA}
}

@article{cantor1991fast,
  title={On fast multiplication of polynomials over arbitrary algebras},
  author={Cantor, David G and Kaltofen, Erich},
  journal={Acta Informatica},
  volume={28},
  number={7},
  pages={693--701},
  year={1991},
  publisher={Citeseer}
}

@article{harvey2022polynomial,
  title={Polynomial multiplication over finite fields in time $O(n \, \text{log}\, n)$},
  author={Harvey, David and Van Der Hoeven, Joris},
  journal={Journal of the ACM (JACM)},
  volume={69},
  number={2},
  pages={1--40},
  year={2022},
  publisher={ACM New York, NY}
}

@article{schonhage1977schnelle,
  title={Schnelle Multiplikation von polynomen {\"u}ber K{\"o}rpern der Charakteristik 2},
  author={Sch{\"o}nhage, Arnold},
  journal={Acta Informatica},
  volume={7},
  number={4},
  pages={395--398},
  year={1977},
  publisher={Springer}
}

@article{schonhage1971schnelle,
  title={Schnelle Multiplikation gro\ss{}er Zahlen},
  author={Sch{\"o}nhage, Arnold and Strassen, Volker},
  journal={Computing},
  volume={7},
  number={3},
  pages={281--292},
  year={1971},
  publisher={Springer}
}

@article{Schost03,
  author    = {Éric Schost},
  title     = {Computing parametric geometric resolutions},
  journal   = {Applicable Algebra in Engineering, Communication and Computing},
  volume    = {13},
  number    = {5},
  pages     = {349--393},
  year      = {2003},
  doi       = {10.1007/s00200-002-0109-x}
}

@inproceedings{vu2025computing,
  title={Computing polynomial representation in subrings of multivariate polynomial rings},
  author={Vu, Thi Xuan},
  booktitle={Proceedings of the 2025 International Symposium on Symbolic and Algebraic Computation},
  pages={284--292},
  year={2025}
}

@article{jeronimo2009deformation,
  title={Deformation techniques for sparse systems},
  author={Jeronimo, Gabriela and Matera, Guillermo and Solerno, Pablo and Waissbein, Ariel},
  journal={Foundations of Computational Mathematics},
  volume={9},
  number={1},
  pages={1--50},
  year={2009},
  publisher={Springer}
}

@article{GiustiLecerfSalvy2001,
  title={A Gr{\"o}bner free alternative for polynomial system solving},
  author={Giusti, Marc and Lecerf, Gr{\'e}goire and Salvy, Bruno},
  journal={Journal of complexity},
  volume={17},
  number={1},
  pages={154--211},
  year={2001},
  doi= {https://doi.org/10.1006/jcom.2000.0571},
  publisher={Elsevier}
}

@article{van2026towards,
  title={Towards a library for straight-line programs},
  author={van der Hoeven, Joris and Lecerf, Gr{\'e}goire},
  journal={Applicable Algebra in Engineering, Communication and Computing},
  volume={37},
  number={2},
  pages={331--387},
  year={2026},
  publisher={Springer}
}

@phdthesis{vu2020homotopy,
  title={Homotopy algorithms for solving structured determinantal systems},
  author={Vu, Thi Xuan},
  year={2020},
  school={Sorbonne Universit{\'e} (France); University of Waterloo (Canada)}
}

@article{el2018bit,
  title={Bit complexity for multi-homogeneous polynomial system solving-application to polynomial minimization},
  author={Safey El Din, Mohab  and Schost, {\'E}ric},
  journal={Journal of Symbolic Computation},
  volume={87},
  pages={176--206},
  year={2018},
  publisher={Elsevier}
}

@article{giusti1997lower,
  title={Lower bounds for Diophantine approximations},
  author={Giusti, Marc and Heintz, J and H{\"a}gele, K and Morais, Jose E and Pardo, LM and Montana, JL},
  journal={Journal of Pure and Applied Algebra},
  volume={117},
  pages={277--317},
  year={1997},
  publisher={Elsevier}
}

@inproceedings{giusti1995polynomial,
  title={When polynomial equation systems can be “solved” fast?},
  author={Giusti, Marc and Heintz, Joos and Morais, Jose Enrique and Pardo, Luis Miguel},
  booktitle={International Symposium on Applied Algebra, Algebraic Algorithms, and Error-Correcting Codes},
  pages={205--231},
  year={1995},
  organization={Springer}
}

@article{kaltofen1989factorization,
  title={Factorization of polynomials given by straight-line programs.},
  author={Kaltofen, Erich},
  journal={Adv. Comput. Res.},
  volume={5},
  pages={375--412},
  year={1989}
}

@article{giusti1998straight,
  title={Straight-line programs in geometric elimination theory},
  author={Giusti, Marc and Heintz, Joos and Morais, Jose Enrique and Morgenstem, Jacques and Pardo, Luis Miguel},
  journal={Journal of pure and applied algebra},
  volume={124},
  number={1-3},
  pages={101--146},
  year={1998},
  publisher={Elsevier}
}

@inproceedings{heintz1981absolute,
  title={Absolute primality of polynomials is decidable in random polynomial time in the number of variables},
  author={Heintz, Joos and Sieveking, Malte},
  booktitle={International Colloquium on Automata, Languages, and Programming},
  pages={16--28},
  year={1981},
  organization={Springer}
}

@article{kaltofen1988greatest,
  title={Greatest common divisors of polynomials given by straight-line programs},
  author={Kaltofen, Erich},
  journal={Journal of the ACM (JACM)},
  volume={35},
  number={1},
  pages={231--264},
  year={1988},
  publisher={ACM New York, NY, USA}
}

@incollection{krick1996computational,
  title        = {A computational method for diophantine approximation},
 author={Krick, Teresa and Pardo, Luis M},
  booktitle    = {Algorithms in Algebraic Geometry and Applications},
  editor       = {Faugère, Jean-Charles and Galligo, André},
  pages        = {193--253},
  year         = {1996},
  publisher    = {Springer},
  address      = {Basel},
  series       = {Progress in Mathematics},
  volume       = {143}
}

@article{hauenstein2021solving,
  title={Solving determinantal systems using homotopy techniques},
  author={Hauenstein, Jon D and El Din, Mohab Safey and Schost, {\'E}ric and Vu, Thi Xuan},
  journal={Journal of Symbolic Computation},
  volume={104},
  pages={754--804},
  year={2021},
  publisher={Elsevier}
}

@article{SafeySchost2017,
author = {Din, Mohab Safey El and Schost, \'{E}ric},
title = {A nearly optimal algorithm for deciding connectivity queries in smooth and nounded real algebraic sets},
year = {2017},
publisher = {Association for Computing Machinery},
address = {New York, NY, USA},
volume = {63},
number = {6},
issn = {0004-5411},
url = {https://doi.org/10.1145/2996450},
doi = {10.1145/2996450},
journal = {Journal of the ACM},
articleno = {48},
numpages = {37},
}

@article{rouillier2000finding,
  title={Finding at least one point in each connected component of a real algebraic set defined by a single equation},
  author={Rouillier, Fabrice and Roy, Marie-Francoise and Safey El Din, Mohab},
  journal={Journal of Complexity},
  volume={16},
  number={4},
  pages={716--750},
  year={2000},
  publisher={Elsevier}
}

@article{rouillier1999solving,
  title={Solving zero-dimensional systems through the rational univariate representation},
  author={Rouillier, Fabrice},
  journal={Applicable Algebra in Engineering, Communication and Computing},
  volume={9},
  number={5},
  pages={433--461},
  year={1999},
  publisher={Springer}
}

@book{vonzurGathenGerhard2013,
  title={Modern computer algebra},
  author={Von Zur Gathen, Joachim and Gerhard, J{\"u}rgen},
  year={2013},
  publisher={Cambridge university press}
}

@article{chevalley1955invariants,
  title={Invariants of finite groups generated by reflections},
  author={Chevalley, Claude},
  journal={American Journal of Mathematics},
  volume={77},
  number={4},
  pages={778--782},
  year={1955},
  publisher={JSTOR}
}

@book{derksen2002computational,
  author    = {Harm Derksen and Gregor Kemper},
  title     = {Computational Invariant Theory},
  series    = {Encyclopaedia of Mathematical Sciences},
  volume    = {130},
  publisher = {Springer},
  year      = {2002}
}

@incollection{alonso1996zeros,
  title={Zeros, multiplicities, and idempotents for zero-dimensional systems},
  author={Alonso, María Emilia and Becker, Eberhard and Roy,  Marie-Fran{\c{c}}ois  and W{\"o}rmann, Thorsten},
  booktitle={Algorithms in algebraic geometry and applications},
  pages={1--15},
  year={1996},
  publisher={Springer}
}

@book{sturmfels1993algorithms,
  title={Algorithms in Invariant Theory},
  author={Sturmfels, Bernd},
  year={1993},
  publisher={Springer-Verlag},
  series={Texts and Monographs in Symbolic Computation},
  address={Vienna, Austria}
}

@inproceedings{basu2017efficient,
  title={Efficient algorithms for computing the euler-poincar{\'e} characteristic of symmetric semi-algebraic sets},
  author={Basu, Saugata and Riener, Cordian},
  booktitle={Ordered Algebraic Structures and Related Topics: International Conference on Ordered Algebraic Structures and Related Topics, October 12--16, 2015, Centre International de Rencontres Math{\'e}matiques (CIRM), Luminy, France},
  volume={697},
  pages={53--81},
  year={2017},
  organization={American Mathematical Soc. Providence, Rhode Island}
}

@article{riener2013exploiting,
  title={Exploiting symmetries in SDP-relaxations for polynomial optimization},
  author={Riener, Cordian and Theobald, Thorsten and Andr{\'e}n, Lina Jansson and Lasserre, Jean B},
  journal={Mathematics of Operations Research},
  volume={38},
  number={1},
  pages={122--141},
  year={2013},
  publisher={INFORMS}
}

@inproceedings{labahn2023faster,
  title={Faster real root decision algorithm for symmetric polynomials},
  author={Labahn, George and Riener, Cordian and Safey El Din, Mohab and Schost, {\'E}ric and Vu, Thi Xuan},
  booktitle={Proceedings of the 2023 International Symposium on Symbolic and Algebraic Computation},
  pages={452--460},
  year={2023}
}

@article{riener2025deciding,
  title={Deciding connectivity in symmetric semi-algebraic sets},
  author={Riener, Cordian and Schabert, Robin and Vu, Thi Xuan},
  journal={arXiv preprint arXiv:2503.12275},
  year={2025}
}

@inproceedings{riener2024connectivity,
  title={Connectivity in symmetric semi-algebraic sets},
  author={Riener, Cordian and Schabert, Robin and Vu, Thi Xuan},
  booktitle={Proceedings of the 2024 International Symposium on Symbolic and Algebraic Computation},
  pages={162--169},
  year={2024}
}

@inproceedings{vu2022computing,
  title={Computing critical points for algebraic systems defined by hyperoctahedral invariant polynomials},
  author={Vu, Thi Xuan},
  booktitle={Proceedings of the 2022 International Symposium on Symbolic and Algebraic Computation},
  pages={167--175},
  year={2022}
}

@article{faugere2023computing,
  title={Computing critical points for invariant algebraic systems},
  author={Faugère, Jean-Charles and Labahn, George and El Din, Mohab Safey and Schost, {\'E}ric and Vu, Thi Xuan},
  journal={Journal of Symbolic Computation},
  volume={116},
  pages={365--399},
  year={2023},
  publisher={Elsevier}
}

@book{lehrer2009unitary,
  title={Unitary reflection groups},
  author={Lehrer, Gustav I and Taylor, Donald E},
  volume={20},
  year={2009},
  publisher={Cambridge University Press}
}

@book{sturmfels2002solving,
  author    = {Bernd Sturmfels},
  title     = {Solving Systems of Polynomial Equations},
  series    = {CBMS Regional Conference Series in Mathematics},
  volume    = {97},
  publisher = {American Mathematical Society},
  year      = {2002}
}

@book{garey2002computers,
  title={Computers and intractability},
  author={Garey, Michael R and Johnson, David S},
  volume={29},
  year={2002},
  publisher={wh freeman New York}
}

@article{fraenkel1979complexity,
  title={Complexity of problems in games, graphs and algebraic equations},
  author={Fraenkel, Aviezri S and Yesha, Yaacov},
  journal={Discrete Applied Mathematics},
  volume={1},
  number={1-2},
  pages={15--30},
  year={1979},
  publisher={Elsevier}
}


\end{document}